\documentclass[runningheads]{llncs}

\usepackage{makeidx}
\usepackage{epsfig}
\usepackage{amsmath}
\usepackage{graphicx}
\usepackage{epstopdf}
\usepackage{url,epsfig,amssymb}

\usepackage[algo2e]{algorithm2e}

\newtheorem{fact}{Fact}

\def\R{\Bbb R}
\def\Z{\Bbb Z}

\begin{document}

\pagestyle{headings}

\mainmatter

\title{Optimal conditions for connectedness\\ of discretized sets}

\titlerunning{Optimal conditions for connectedness of discretized sets}

\author{Boris Brimkov \and Valentin E. Brimkov}
\authorrunning{Boris Brimkov and Valentin E. Brimkov}   
%
\institute{Department of Computational \& Applied Mathematics,\\
Rice University, Houston, TX 77005, USA\\
\email{boris.brimkov@rice.edu}\\
\and
Mathematics Department,\\
SUNY Buffalo State, Buffalo, NY 14222, USA\\
\email{brimkove@buffalostate.edu}}

\maketitle

\begin{abstract}
Constructing a discretization of a given set  is a major problem in various theoretical and applied disciplines.
An offset discretization of a set $X$ is obtained by taking the integer points inside a closed neighborhood of $X$ of a certain radius. 
In this note 
we determine a minimum threshold for the offset radius, beyond which  the discretization of a disconnected set is always connected. The results hold for a broad class of  disconnected and unbounded subsets of $\R^n$, and generalize several previous results. Algorithmic aspects and possible applications are briefly discussed. 

  
\end{abstract}
\section{Introduction}
Constructing a discretization of a set $X \subseteq \R^n$ is a major problem in various theoretical and applied disciplines, such as numerical analysis, discrete geometry, computer graphics, medical imaging, and image processing. For example, in numerical analysis, one may need to transform a continuous domain of a function into its adequate discrete analogue. In raster/volume graphics,  one looks for  a rasterization that converts an image described in a vector graphics format into a raster image built by pixels or voxels.  Such studies often elucidate interesting relations between continuous structures and their discrete counterparts.   

Some of the earliest ideas and results for set discretization belong to Gauss (see, e.g., \cite{gauss}); Gauss discretization is still widely used in theoretical research and applications. 
A number of other types of discretization have been studied by  a large number of authors (see, e.g., \cite{andres,andres2,Kauf97,figueiredo96new,jonas,kaufman1,Kim83,rosenfeld,rosenfeld2,tajine} 
and the bibliographies therein). These works focus on special types of sets to be discretized, such as straight line segments, circles, ellipses, or some other classes of curves in the plane or on other surfaces.  

An important requirement for any discretization is to preserve certain topological properties of the original object. Perhaps the most important among these is the connectedness or disconnectedness of the discrete set obtained from a discretization process.  
This may be crucial for various applications ranging from medicine and bioinformatics  (e.g. organ and tumor measurements in CT images, beating heart or lung simulations, protein binding simulations) to robotics and engineering (e.g. motion planning, finite element stress simulations).
Most of the works cited above address issues related to the connectedness of the obtained discretizations.  To be able to perform a reliable study of the topology of a digital object by means of shrinking, thinning and skeletonization algorithms (see, e.g. \cite{gabriela0,gabriela1,KLE2004,gabriela3,gabriela2} and the bibliography therein), one needs to start from a faithful digitization of the original continuous set.

Perhaps the most natural and simple type of discretization of a set $X$ is the one  defined by the integer points within a closed neighborhood of $X$ of a certain radius $r$.
This will be referred to as an $r$-offset discretization.  
Several authors have studied properties of offsets of certain curves and surfaces \cite{offset2,offset3,offset4,offset5}, however without being concerned with the properties of the integer set of points enclosed within the offset. In
\cite{Jamet06b,Jamet06a} results are presented on offset-like conics discretizations.   
Conditions for connectedness of offset discretizations of bounded path connected or connected sets are presented in \cite{brimkov,bbb,stelldinger}. 

While all related works study conditions under which connectedness of the original set is preserved upon discretization, in the present paper we determine minimum thresholds for the offset radius, beyond which disconnectedness of a given original set is never  preserved, i.e., the obtained discretization is always connected. 
The results hold for a broad class of  disconnected subsets of $\R^n$, which are allowed to be unbounded. 
The technique we use is quantizing the (possibly uncountable and unbounded) set $X$ by the minimal countable (possibly infinite) set of voxels containing $X$, which makes the use of induction feasible.   
To our knowledge, these are the first results concerning offset discretizations of disconnected sets. They extend a result from \cite{bbb} which gives best possible bounds for an offset radius to  guarantee 0- and $(n-1)$-connectedness of the offset discretization of a bounded path-connected set; they also generalize a result from \cite{brimkov} to unbounded connected sets.

In the next section we introduce various notions and notations to be used in the sequel. 
In Sections \ref{main1} we present the main results of the paper.
We conclude in Section~\ref{concl} by addressing some algorithmic aspects and possible applications.  
\section{Preliminaries}
\label{sec:def}
We recall a few basic notions of general topology and discrete geometry. For more details we refer to \cite{engelking,KLE2004,KON2001}. 

All considerations take place in $\R^n$ with the Euclidean norm.
By $d(x,y)$ we denote the {\em Euclidean distance} between points $x,y \in \R^n$. Given two sets $A,B \subset \R^n$, the number $g(A,B) = \inf_{x,y} \{ d(x,y) : x \in A, y \in B \}$ is called the {\em gap}\footnote{The function $g$ itself, defined on the subsets of $\R^n$ is called a {\em gap functional}. See, e.g., \cite{beer} for more details.} between the sets $A$ and $B$.
$B^n(x,r)$  is the closed $n$-ball of radius $r$ and center $x$ (dependence on $n$ will be omitted when it is clear from the context).
Given a set $X \subseteq \R^n$, $|X|$ is its cardinality. The {\em closed $r$-neighborhood} of $X$, which we will also refer to as the $r${\em -offset} of $X$, 
is defined by $U(X,r)=\cup_{x \in X} B(x,r)$.
$Cl(X)$ is the {\em closure} of $X$, i.e., the union of $X$ and the limit points of $X$. 
$X$ is {\em connected} if it cannot be presented as a union of two nonempty subsets 
that are contained in two disjoint open sets.
Equivalently, $X$ is connected if and only if it cannot be presented as a union of two nonempty subsets each of which 
is disjoint from a closed superset of the other.

In a discrete geometry setting, considerations take place in the {\em grid cell model}.
In this model, the regular orthogonal grid subdivides $\R^n$ into
$n$-dimensional unit hypercubes (e.g., unit squares for $n=2$ or unit cubes for $n=3$).
These are regarded as {\it $n$-cells}  and are called {\em hypervoxels}, or {\em voxels}, for short.
The $(n-1)$-cells, $1$-cells, and $0$-cells of a voxel are referred to as {\em facets}, {\em edges},  and {\em vertices}, respectively. 

Given a set $X \subseteq \R^n$, $X_{\Z}= X \cap \Z^n$ is its {\em Gauss discretization}, while
$\Delta_r(X)=U(X,r) \cap \Z^n$ is its {\em discretization of radius $r$}, which we will also call the {\em $r$-offset discretization} of $X$.

Two integer points are $k$-{\em adjacent} for some $k$,
$0 \leq k \leq n-1$, iff no more than $n-k$ of their coordinates differ by 1. A $k$-{\em path} (where $0 \leq k \leq n-1$) in a set $S \subset \Z^n$ is a sequence of
integer points from $S$ such that every two consecutive points of the path
are $k$-adjacent. 
Two points of $S$ are {\em $k$-connected} (in $S$) iff there is a $k$-path in $S$ between them. 
$S$ is $k$-{\em connected} iff there is a $k$-path in $S$ connecting any two points of $S$. 
If $S$ is not $k$-connected, we say that it is $k$-{\em disconnected}.
A maximal (by inclusion) $k$-connected subset of $S$ is called a $k$-{\em (connected) component} of $S$.
Components of nonempty sets are nonempty and any union of distinct $k$-components is $k$-disconnected.
Two voxels $v,v'$ are $k$-adjacent if they share a $k$-cell.
Definitions of connectedness and components of a set of voxels are analogous to those for integer points.

In the proof of our result we will use the following well-known facts (see \cite{bbb}). 
\begin{fact}
\label{L1}
Any closed $n$-ball $B \subset \R^n$ with a radius greater than or equal to $\sqrt{n}/2$ 
contains at least one integer point.
\end{fact}
\begin{fact}
\label{F1}
Let $A$ and $B$ be sets of integer points, each of which is $k$-connected.
If there are points $p \in A$ and $q \in B$ that are $k$-adjacent, then $A \cup B$ is $k$-connected.
\end{fact}
\begin{fact}
\label{F11}
If $A$ and $B$ are sets of integer points, each of which is $k$-connected, and $A \cap B \neq \emptyset$,
then $A \cup B$ is $k$-connected.
\end{fact}
\begin{fact}
\label{L2}
Given a closed $n$-ball $B \subset \R^n$ with $B_{ \Z} \neq \emptyset$, $B_{ \Z}$ is $(n-1)$-connected.
\end{fact}

\section{Main Result}
\label{main1}
In this section we prove the following theorem.
\begin{theorem}
Let $X \subset \R^n$, $n \geq 2$, be a disconnected set such that $Cl(X)$ is connected. Then the following hold:  
\begin{enumerate}
\item[1.]
$\Delta_r(X)$ is   $(n-1)$-connected for all $r > \sqrt{n}/2$.
\item[2.]
$\Delta_r(X)$ is at least 0-connected for all $r > \sqrt{n-1}/2$.
\end{enumerate}
These bounds are the best possible which  always respectively guarantee
$(n-1)$ and 0 connectedness of $\Delta_r(X)$.
\label{th1}
\end{theorem}
\begin{proof}
The proof of the theorem is based on the following fact.
\begin{claim}
\label{Lnew}
Let $X \subseteq \R^n$ be an arbitrary disconnected set (possibly infinite), such that $Cl(X)$ is connected. 
Let $W(X)$ denote the (possibly infinite) set of voxels intersected by $X$.  
Then $W(X)$ can be ordered in a sequence $W(X)= \{ v_1, v_2, \dots \}$ with the following property:
\begin{equation}
Cl(X) \cap \left(v_k \cap \bigcup_{i = 1}^{k-1} v_i\right) \neq \emptyset, \ 
\forall k \geq 1.
\label{property1}
\end{equation} 
\end{claim}
\begin{proof}
To simplify the notation, let $\bigcup F$ stand for the union of a family of sets $F$. Let $W'(X)$ be a maximal by inclusion subset of $W(X)$ satisfying Property~(\ref{property1}). Note that $W'(X)$ always exists, no matter if $W(X)$ is finite or infinite.
Assume for contradiction that $W'(X) \neq W(X)$.
By the maximality of $W'(X)$ it follows that $X_1:=Cl(X) \cap \bigcup W(X)$ 
does not intersect the closed set $Y_1:=\bigcup (W(X) \setminus W'(X))$, and $X_2:=Cl(X) \cap \bigcup (W(X) \setminus W'(X))$ does not intersect the closed set $Y_2:=\bigcup W'(X)$.
Then we have that $Cl(X)$ is  the union of the nonempty sets $X_1$ and   
$X_2$, and each of them is disjoint from a closed superset of the other ($Y_1$ and $Y_2$, respectively), which is impossible if $Cl(X)$ is connected. \qed
\end{proof}

For the proof of both parts of the theorem we use induction on $k$ to establish the claimed connectedness of $\mathit{\Delta}_r(X \cap \bigcup_{i=1}^k v_i)$.

{\em Part 1.} \ \ Let $W(X)= \{ v_1, v_2, \dots \}$ be defined as in Claim~\ref{Lnew}, with a voxel ordering satisfying Property~(\ref{property1}).  

Let  $k=1$, i.e. the set $W(X)$ consists of a single voxel $v$ and 
$X \subseteq v$. 
By Fact~\ref{L1} we have that $\Delta_r(X) \neq \emptyset$. 
Denote for brevity $D=\Delta_r(X)$ and assume for contradiction that $D$ has at least two $(n-1)$-connected components. 
Let $D_1$ be one of these components. 
Denote $D_2 = D \setminus D_1 \neq \emptyset$ and define the sets 
$A_1 = \bigcup_{p \in D_1} B(p,r)$ and $A_2 = \bigcup_{p \in D_2} B(p,r)$.
Then we have that $X \subseteq A_1 \cup A_2$. 
To see why, assume that there is $x \in X$, $x \notin A_1 \cup A_2$.
By Fact~\ref{L1},   $B(x,r)$ contains an integer point $q$.
Then $q \in D$ and $d(x,q) \leq r$. 
Then $x \in B(q,r) \subseteq \bigcup_{p \in D} B(p,r) =  A_1 \cup A_2$, a contradiction. 

Since $X$ is bounded, it follows that $D$ is finite, and therefore $A_1$, $A_2$, and $A_1 \cup A_2$ are closed.  
Hence, 
$$
Cl(X) \subseteq Cl(A_1 \cup A_2)=A_1 \cup A_2.
$$
Next, we observe that $Cl(X) \cap A_1 \cap A_2 \neq \emptyset$.
Otherwise, we would have 
$$
Cl(X) = (Cl(X) \cap A_1) \cup (Cl(X) \cap A_2),
$$
where 
$Cl(X) \cap A_1 \neq \emptyset$ and  
$Cl(X) \cap A_2 \neq \emptyset$, 
as $Cl(X) \cap A_1$ is disjoint from $A_2$ and  
$Cl(X) \cap A_2$ is disjoint from $A_1$; thus $Cl(X)$ would be a union of two nonempty sets, each of which is disjoint from a closed set containing the other (since $A_1$ and $A_2$ are closed), which contradicts the connectedness of $Cl(X)$.
Then there exist points $p_1 \in D_1$, $p_2 \in D_2$, such that 
$$
Cl(X) \cap  B(p_1,r) \cap  B(p_2,r) \neq \emptyset.
$$
Hence, there is a point $a \in Cl(X)$, such that $a \in  B(p_1,r) \cap  B(p_2,r)$, which is possible only if   
$d(a,p_1) \leq r$ and $d(a,p_2) \leq r$.  
Thus it follows that $p_1,p_2 \in  B(a,r)$. 
Since $a$ is a limit point of $X$, there is a point $b \in X$,
such that $p_1,p_2 \in  B(b,r)$, provided that the radius $r$ is strictly greater than $\sqrt{n}/2$. 
Since $b \in X$, we also have that $B(b,r)_{\Z} \subseteq D$. 
Then Fact~\ref{L2} implies that points $p_1$ and $p_2$ are $(n-1)$-connected in $D$, which contradicts the assumption that 
$D_1$ is an $(n-1)$-connected component of $D$ with $D_1 \neq D$.

Now suppose that $\mathit{\Delta}_r(X \cap \bigcup_{i=1}^{k} v_i)$ is $(n-1)$-connected for some $k \geq 1$. 
We have  
$$
X \cap \bigcup_{i=1}^{k+1} v_i = (X \cap v_{k+1}) \cup (X \cap \bigcup_{i=1}^{k} v_i).
$$
Since the closed $r$-neighborhood of a union of two sets equals the union of their $r$-neighborhoods, it follows that
$$
\mathit{\Delta}_r(X \cap \bigcup_{i=1}^{k+1} v_i) = \mathit{\Delta}_r(X \cap v_{k+1}) \cup \mathit{\Delta}_r(X \cap \bigcup_{i=1}^{k} v_i).
$$
Let us denote by $f$ a common face of voxel $v_{k+1} \in W(X)$ and the polyhedral complex composed by the voxels $v_1,v_2,\dots,v_k$, i.e., $F \subset v_{k+1} \cap \bigcup_{i=1}^{k} v_i$.
W.l.o.g., we can consider the case where $f$ is a facet of $v_{k+1}$ (i.e., a cell of topological dimension $n-1$), the cases of lower dimension faces being analogous. 
Let $H$ be the hyperplane in $\R^n$ which is the affine hull of $f$. 
Let $\Z^{n-1}_{H}$ be the subset  of the set of grid-points $\Z^n$
contained in $H$.

By Claim~\ref{Lnew}, there is a point $p \in Cl(X) \cap f$. 
Consider the $n$-ball $B^n(p,r)$.  
Then $B^{n-1}(p,r) = B^n(p,r) \cap H$ is an $(n-1)$-ball with the same center and radius.
Applying Fact~\ref{L1}  to $B^{n-1}(p,r)$ in the $(n-1)$-dimensional hyperplane $H$,   
we obtain that $B^{n-1}(p,r)$ contains at least one grid point 
$q \in \Z^{n-1}_{H}$, which is a vertex of facet $f$.
Since $p \in Cl(x)$ is a limit point of $X$, 
there exists a point $p' \in X$, such that the ball $B^{n-1}(p',r)$ contains $q$, too.
By construction, $q$ is common for the sets $\mathit{\Delta}_r(X \cap v_{k+1})$ and $\mathit{\Delta}_r(X \cap \bigcup_{i=1}^{k} v_i)$.
The former is $(n-1)$-connected by the same argument used in the induction basis, while the latter is $(n-1)$-connected by  the induction hypothesis. 
Then by Fact~\ref{F11}, their union $\mathit{\Delta}_r(X \cap \bigcup_{i=1}^{k+1} v_i)$ is $(n-1)$-connected, as well.
This establishes Part 1.

\smallskip

{\em Part 2.} \ \ The proof of this part is similar to the one of Part 1. 
Note that in the base case $k=1$ (i.e. when $X$ is contained in a single voxel $v$), if $\sqrt{n-1}/2 < r < \sqrt{n}/2$ then it is possible to have 
$\Delta_r(X) = \emptyset$ (e.g., if $X$ consists of a single point that is the center of $v$). 
If that is the case, the statement follows immediately.  
Thus, suppose that $\Delta_r(X) \neq \emptyset$. 
By definition, $\Delta_r(X) = \bigcup_{x\in X} B(x,r)_{\Z}$.
For any $x \in X$,  $B(x,r)_{\Z}$ is $(n-1)$-connected by Fact~\ref{L2}. 
We also have that any of the nonempty sets $B(x,r)_{\Z}$ contains  a vertex of $v$.
Since any two vertices of a grid cube are at least 0-adjacent, it follows that any subset of vertices of $v$ is at least 0-connected. 
Then by Fact~\ref{F1}, $\Delta_r(X)$ is at least 0-connected.
 
The rest of the proof parallels the one of Part 1, with the only difference that 
point $q$ is common for the sets $\mathit{\Delta}_r(X \cap v_{k+1})$ and $\mathit{\Delta}_r(X \cap \bigcup_{i=1}^{k} v_i)$, each of which is 0-connected (the former by an argument used in the induction basis, and  the latter by  the induction hypothesis). 
Then Fact~\ref{F11} implies that their union $\mathit{\Delta}_r(X \cap \bigcup_{i=1}^{k+1} v_i)$ is 0-connected, as stated.

Figure~\ref{fig} illustrates that the obtained bounds for $r$ are 
the best possible: if $r$ equals $\sqrt{n}/2$ (resp. $\sqrt{n-1}/2$), then 
$\Delta_r(X)$ may not be $(n-1)$-connected (resp. 0-connected).
This completes the proof of the theorem. \qed
\end{proof}
\begin{figure}[h!]
\begin{center}
\begin{tabular}{c c c}
\includegraphics[clip=true, width=37.5mm, clip=true]{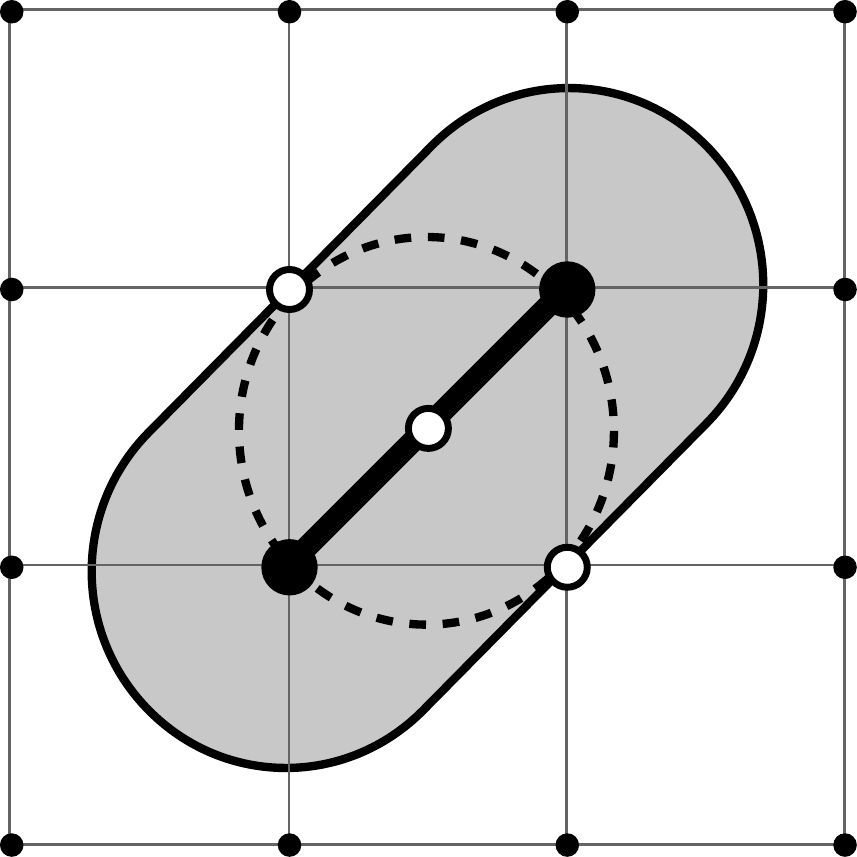} & 
\hspace{10mm} &
\includegraphics[clip=true, width=50mm, clip=true]{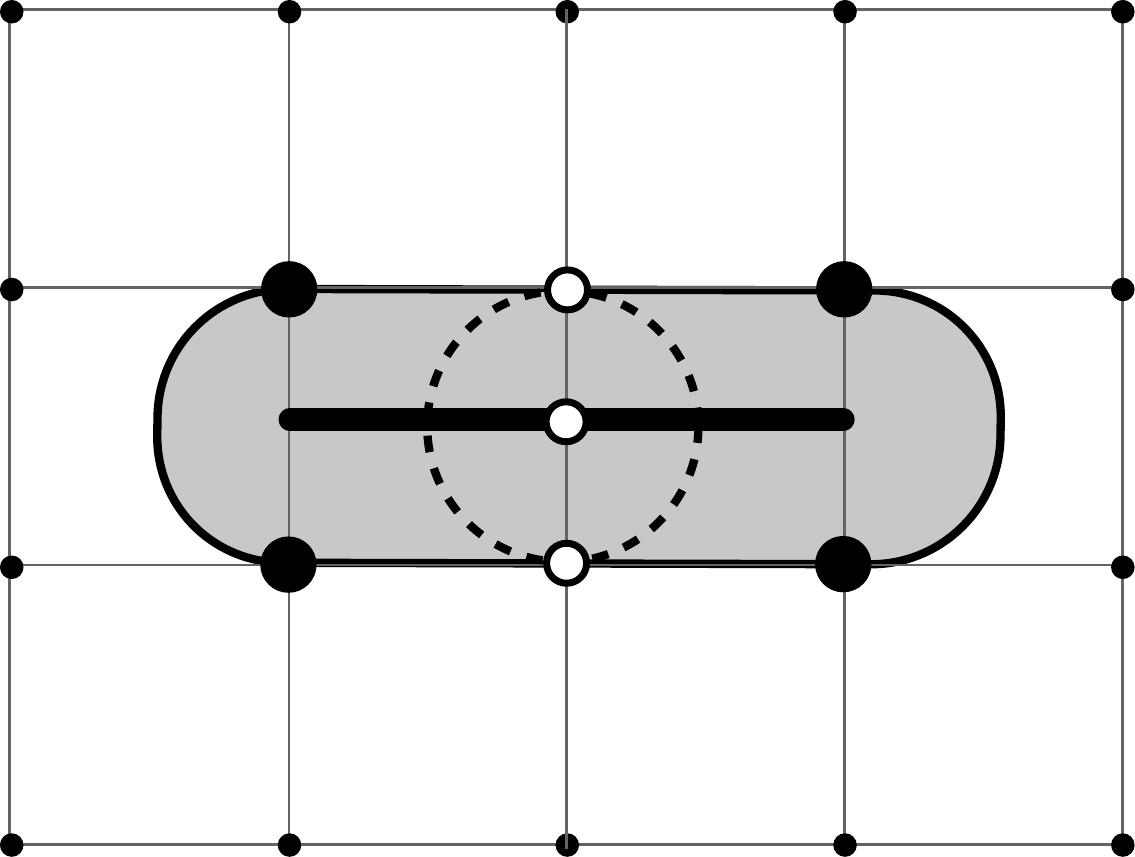} 
\end{tabular}
\end{center}
 \caption{{\em Left:} Offset radius $r=\sqrt{2}/2$, {\em Right:}  offset radius $r=1/2$. In both figures: $X$ is the thick line segment with missing midpoint marked by a hollow dot, the shaded region is the offset, the offset discretization consists of the large thick dots, the hollow dots on the offset boundary do not belong to the offset and to the discretization.} 
\label{fig}
\end{figure}
The proof of Theorem~\ref{th1} implies the following corollary.
\begin{corollary}  
If $X \subseteq \R^n$ ($n \geq 2$) is connected, then  
$\Delta_r(X)$ is   $(n-1)$-connected for all $r \geq \sqrt{n}/2$, and 
$\Delta_r(X)$ is at least 0-connected for all $r \geq \sqrt{n-1}/2$.
\label{cor1}
\end{corollary}
The above in turn implies:
\begin{corollary}  
If $X \subset \R^n$ ($n \geq 2$) is disconnected but $U(X,r)$ is connected for some $r>0$, then  
$\Delta_{r+\sqrt{n}/2}(X)$ is $(n-1)$-connected, and 
$\Delta_{r+\sqrt{n-1}/2}(X)$ is at least 0-connected. 
\label{cor2}
\end{corollary}

\section{Algorithmic aspects and applications}
\label{concl}
 
Let $X$ be a closed disconnected subset of $\mathbb{R}^n$ where $n \geq 2$. 
Denote by $\alpha_j(X)$, $j=0$ or $n-1$, the minimum value of an offset radius    
for which $\Delta_{\alpha_j(X)}(X)$ is $j$-connected. 
Let $\rho(X)$ be the smallest offset radius for which $U(X,r)$ is connected. 
Knowing the exact value of $\rho(X)$ or having a bound on it, one can easily estimate $\alpha_j(X)$ with the help of Corollary~\ref{cor2}. 

Let $X \subseteq \mathbb{R}^n$ be a bounded set with closed components $X_1,X_2,\dots,X_m$ and 
$g_{ij}=g(X_i,X_j)$ be the gap between $X_i$ and $X_j$ for $1 \leq i,j \leq m$. 
It is not hard to see that 
$\rho(X) \leq \frac{1}{2}\delta(X)$, where $\delta(X):= \min_i \max_j g_{ij}$. Given the values $g_{ij}$, $\delta(X)$ can be found in $O(m^2)$ time. Another upper bound on $\rho(X)$ is given by the radius $\omega(X)$ of the minimal bounding sphere for $X$.   
Recall that, given a non-empty family of bounded sets $X$ in $\R^n$, a {\em minimal bounding sphere} for that family is the sphere of minimum radius such that the closed ball bounded by the sphere contains all sets of the family.
Corollary~\ref{cor2} implies that 
$\Delta_{(\delta(X)+\sqrt{n})/2}(X)$ and $\Delta_{\omega(X)+\sqrt{n}/2}(X)$ are $(n-1)$-connected, while $\Delta_{(\delta(X)+\sqrt{n-1})/2}(X)$ and 
$\Delta_{\omega(X)+\sqrt{n-1}/2}(X)$ are $0$-connected. 
For the special case where $X$ is a set of points in $\R^n$, $\omega(X)$ can be computed in linear time $O(m)$ for any fixed dimension $n$  by Megiddo's ``prune and search" minimal bounding sphere algorithm 
\cite{megiddo1,megiddo2}. 
In that case we also have the following relation.   
\begin{proposition}
\label{31}
Let $X$ be a set of $m\geq 2$ points in $\R^n$.
Then $\rho(X) \leq \omega(X) \leq \delta(X)$.
\end{proposition}
\begin{proof}(sketch)\\
Let $S(X)$ be the minimum bounding sphere of $X=\{p_1,p_2,\dots,p_m \}$ with center ${\cal C}$ and radius $R$. If $X=\{ p_1,p_2 \}$, then $p_1$ and $p_2$ are end-points of a diameter of $S(X)$ and $\delta(X) = 2R > \omega(X)=R$. 

Now suppose that $m\geq 3$. Suppose that $S(X)$ contains two points $a$ and $b$ from $X$, which are endpoints of a diameter of $S(X)$. 
Let $c$ be another point from $X$.   
If $c \neq {\cal C}$, then it is easy to see that at least one of the inequalities $d(a,c)>R$ or $d(b,c)>R$ holds.     
If $c={\cal C}$, then $d(a,c)=R$. 
Thus we have that that $R=\omega(X)\leq \delta(X)$, as equality holds if and only if a point from $X$ is a center of $S(X)$.  

Now consider the case where $S(X)$ does not contain endpoints of a diameter of $S(X)$. 
Then $S(X)$ contains a set $M$ of at least three points, such that the center ${\cal C}$ belongs to the  convex hull $conv(M)$ of $M$ (otherwise   the sphere $S(M)$ would not be minimum enclosing for $X$).
Suppose that $X$ contains no point at the center ${\cal C}$. 
Then there are points $a, b \in M$ with $d(a,b) > R$, which once again implies $\delta(X) > R = \omega(X)$. To see why, assume for contradiction that the diameter $diam(M)$ of $M$ satisfies $diam(M) \leq R$. Then, if $diam(M) < R$, the polytope $conv(M)$ cannot contain ${\cal C}$ which is at a distance $R$ from any of its vertices that are elements of $M$. If $diam(M)=R$, then ${\cal C}$ must be among the vertices of $conv(M)$, since $diam(M)$ equals the diameter of $conv(M)$ and is achieved for a pair of its vertices.
If $X$ contains a point that coincides with ${\cal C}$, then we clearly have $\delta(X) = R = \omega(X)$. \qed
\end{proof}

It was shown in \cite{bbb} that, given an array of gap values $A(X)=\{ g_{ij}, \ i,j=1,2,\dots, m \}$, $\rho(X)$ can be computed in $O(m^4)$ time. Here we observe that this can be performed much more efficiently in $O(m^2)$ time by constructing a minimum spanning tree of a complete graph on $m$ vertices, for which the array $A(X)$ is the adjacency matrix of edge weights. This can be done, e.g., by Prim's algorithm \cite{prim} with $O(m^2)$ arithmetic operations.   Then $\rho(X)$ is the value of the maximum edge weight in the obtained spanning tree (recall that the (multi)set of weights is unique for all minimum spanning trees of a graph).

In the Introduction we briefly discussed the theoretical and practical worth of results like those presented in this article. 
We conclude by adding one more comment.  
Suppose that a connected set $X$ (e.g., a continuous image) to be discretized is partially ``flawed" and made noisy by discarding some isolated points or lines from the image, whose removal makes it disconnected. Nonetheless, Theorem~\ref{th1} guarantees that one can get a faithful connected digitization of $X$ by choosing an offset size specified by the theorem.

In this note we obtained theoretical conditions for connectedness of offset discretizations of sets in higher dimensions. An important future task is seen in computer implementation and testing the topological properties and visual appearance of offset discretizations of varying  radius. It would also be interesting to study similar properties of other basic types of discretization. 


\end{document}